\documentclass[conference]{IEEEtran}

\usepackage{amssymb, amsmath}
\usepackage{amsthm}
\usepackage{graphicx}  

\newcommand{\nop}[1]{}

\newtheorem{definition}{Definition}

\newtheorem{theorem}{Theorem}
\newtheorem{corollary}{Corollary}
\newtheorem{lemma}{Lemma}
\newtheorem{proposition}{Proposition}

\begin{document}

\title{Stochastic Service Curve and Delay Bound Analysis: A Single Node Case}         
\author{Yuming Jiang \\Norwegian University of Science and Technology (NTNU),~Norway}        
\date{\today}          
\maketitle

\begin{abstract}
A packet-switched network node with constant capacity (in bps) is considered, where packets within each flow are served in the first in first out (FIFO) manner. While this single node system is perhaps the simplest computer communication system, its stochastic service curve characterization and independent case analysis in the context of stochastic network calculus (snetcal) are still basic and many crucial questions surprisingly remain open. Specifically, when the input is a single flow, what stochastic service curve and delay bound does the node provide? When the considered flow shares the node with another flow, what stochastic service curve and delay bound does the node provide to the considered flow, and if the two flows are independent, can this independence be made use of and how? The aim of this paper is to provide answers to these fundamental questions. 
\end{abstract}

\section{Introduction} \label{sec-1}

Network calculus is a theory dealing with queueing type problems encountered in packet-switched computer networks. To simplify the analysis, an important idea in network calculus is to characterize the traffic and service processes using some bounds and perform analysis based on such bounds. Network calculus has developed along two tracks --- deterministic and stochastic. Deterministic network calculus, coined by \cite{Cruz91ab}, has been extensively studied since its introduction in early 1990s, and is nicely covered by two books \cite{Chang00} \cite{NetCal}. Stochastic network calculus is the probabilistic extension or generalization of deterministic network calculus. The development of {\em stochastic network calculus} (SNC) began also in early 1990s. Early representative works include \cite{Kurose92}\cite{YS93}\cite{Chang94} for traffic modeling, and \cite{Lee95} for server modeling. The book \cite{Chang00} also covers the theory of effective bandwidth, a first approach to SNC. However, due to challenges specific to stochastic networks, it is recently that crucial network calculus properties have been proved for SNC, e.g. \cite{Burchard06}\cite{Ciucu06}\cite{Li07}\cite{Sigcomm06}\cite{Liu07}\cite{Jiang-comnet09}\cite{Fidler06b}. A selection of recent results can be found in the book \cite{SNetCal}. In addition, three surveys/overviews are available \cite{Mao06}\cite{Fidler10}\cite{Jiang12}. 

In SNC, stochastic service curve is the fundamental concept for server modeling. If some flows and servers are independent, it is expected that tighter analytical bounds can be obtained by making use of this independence information in the analysis. In this paper, we consider a work-conserving constant capacity node serving flows. Each flow consists of a sequence of packets that are served in the first in first out (FIFO) manner.\nop{ The packet sizes may be independent of the packet arrival times.} This single node system is perhaps the simplest computer communication system. For such a system, an immediate impression is perhaps that it has been thoroughly investigated and is well understood, given the current snetcal literature. Unfortunately, this impression has no solid supporting ground and can hence be highly misleading. Indeed, while the snetcal literature has a lot of results based on its various stochastic arrival curve and stochastic service curve models, the following questions remain largely open\nop{, particularly for cases where independence seems to be obvious or natural}. What stochastic service curve and delay bound does the node provide when the input is a single flow? What stochastic service curve and delay bound does the node provide when the considered flow shares the node with another flow? \nop{If packet arrival times and packet sizes are independent,} If the two flows are independent, can this independence be made use of and how? 

The objective of this paper is to derive results providing answers to these fundamental questions. Specifically, (i) when there is only one flow\nop{ and packet sizes are a stationary process}, we prove a stochastic service curve (SSC) that has a bounding function equal to the complimentary cumulative distribution function (CCDF) of the packet length distribution. In addition to the delay bound directly obtained from the existing snetcal results and this SSC, an improved delay bound is derived, which is consistent with a result in the deterministic network calculus literature: {\em For delay bound analysis, the last packetizer may be ignored}\cite{Chang00}\cite{NetCal}. (ii) When there is cross traffic, i.e., the node is shared by the traversing flow with a crossing flow, we prove that the node provides to the aggregate of the two flows an {\em aggregate behavior} stochastic service curve that also has a bounding function equal to the CCDF of the packet length distribution of the aggregate. Based on this and existing snetcal results, an SSC for the traversing flow is found. To overcome the potential difficulty in finding the packet length CCDF of the aggregate flow, a new and improved SSC is derived, where the flow independence can also be made use of. Moreover, in addition to delay bounds from these SSCs, an improved delay bound is obtained, where\nop{ the packetizer effect is not explicit and} the flow independence information can be exploited\nop{ to improve the obtained bound}. (iii) To illustrate the obtained delay bounds, two examples are provided. For the single flow case, the (best) bound matches with the exact result for $M/M/1$/FIFO. For the case with cross traffic, the obtained (best) bound is close to the exact result for\nop{ priority queue} $M/M/1$/priority. 

The rest is structured as follows. In the next section, the system model and notation are defined. In Sec.~\ref{sec-3}, stochastic network calculus basics are given. In Sec.~\ref{sec-4}, the difficulties for stochastic service curve and delay bound analysis, when packetization effect is not ignored, are discussed. In Sec.~\ref{sec-5}, we focus on the single flow case\nop{ and find stochastic service curve and delay bounds for the flow}. In Sec.~\ref{sec-6}, we take cross traffic into consideration, and find stochastic service curves and delay bounds for the traversing flow. In Sec.~\ref{sec-7}, we give examples\nop{ to demonstrate the use of the derived results}. In Sec.~\ref{sec-8}, discussion on related work is provided. Finally, concluding remarks are given in Sec.~\ref{sec-9}. 

\section{The System Model and Notation} \label{sec-2}

We consider a work-conserving network node serving flows in a packet-switched network. It is a discrete-time system with time indexed by $t=0, 1, 2, \dots$. \nop{The length of the time unit is supposed to be small enough to even count one bit service time, e.g. $\frac{1}{C}$} The serving capacity of the node is constant, denoted by $C$ (in bps). Flow $f$ traverses this node and is referred as the {\em traversing flow}. In addition, the node may also serve another flow $f^{c}$, which is the aggregate flow of crossing traffic and is referred as the {\em crossing flow}. Packets within each flow are served in the FIFO manner. Between flows, some scheduling policy is employed, but within this paper, it is not specified. 

By convention, a packet is said to have arrived to (respectively served by) the node when and only when its last bit has arrived to (respectively left) the node. When a packet arrives seeing the node busy, the packet will be queued and the buffer size for such a queue is assumed to be large enough ensuring no packet loss. All queues are initially empty. 

For the traversing flow $f$, we let $p^{f,i}$ denote the $i$th packet $(i = 1, 2, \dots)$ of the flow. For each $p^{f,i}$, we denote by $a^{f,i}$ its arrival time to the node, $d^{f,i}$ its departure time from the node, and $l^{f,i}$ its length (in bits). 
Similarly, for the crossing flow $f^{c}$, we let $p^{c,i}$ denote its $i$th packet $(i = 1, 2, \dots)$, $a^{c,i}$ its arrival time, $d^{c,i}$ its departure time, and $l^{c,i}$ its length. 

We further use $A^{f}(t)$ and $A^{c}(t)$ to denote the amount of traffic (in bits) that has arrived from the traversing flow and the crossing flow to the node within time period $[0, t]$ respectively. Correspondingly, $A^{f}(s,t) = A^{f}(t)-A^{f}(s)$ and $A^{c}(s,t)=A^{c}(t)-A^{c}(s)$ respectively denote the amount of traffic (in bits) that has arrived from them within time period $(s, t]$. For the departures from the node, we use $A^{*f}(t)$ and $A^{*c}(t)$ to respectively denote the amount of traffic (in bits) that has been served from the traversing flow and the crossing flow within time period $[0, t]$. Correspondingly, $A^{*f}(s,t) = A^{*f}(t)-A^{*f}(s)$ and $A^{*c}(s,t)=A^{*c}(t)-A^{*c}(s)$ respectively represent the amount of traffic (in bits) that has been served from the traversing flow and the crossing flow by the node within time period $(s, t]$. 

For the node, if it is shared by $f$ and $f^{c}$, consider {\em the sequence of packets on the output link}. For this sequence of packets, we call it the {\em aggregate flow} at the node, and let $p^{g,j}$ denote the $j$th packet $(j = 1, 2, \dots)$ of the aggregate flow. For each $p^{g,j}$, denote by $a^{g,j}$ the arrival time of the corresponding packet to the node, $d^{g,j}$ its departure time from the node, and $l^{g,j}$ its length (in bits). Note that the aggregate flow is resulted from the aggregation of the traversing flow and the crossing flow through the work-conserving constant capacity node. In addition, for the departures from the node, we use $A^{*g}(t)$ to  denote the amount of traffic (in bits) from the aggregate flow, which has been served by the node within time period $[0, t]$, and $A^{*g}(s,t)=A^{*g}(t)-A^{*g}(s)$ the amount of traffic (in bits) that has been served from the aggregate flow by the node within time period $(s, t]$. It is worth highlighting that for the departures, there holds $A^{*g}(s,t) = A^{*f}(s,t) + A^{*c}(s,t)$. 

The delay of $p^{f,i}$, denoted by $D^{f,i}$, is naturally 
\begin{eqnarray}
D^{f,i} &=& d^{f,i} - a^{f,i} .
\end{eqnarray}

In addition, we define the (virtual) delay at time $t$ as 
\begin{eqnarray}
D^{f}(t) &=&  \inf \{\tau: A^{*f}(t+\tau) \ge A^{f}(t) \} .
\end{eqnarray}
Due to FIFO, the delay of $p^{f,i}$ is also found from\footnote{Strictly speaking, $D^{f,i} \le D^{f}(a^{f,i})$\nop{ should be used}, where the equation holds only if there is no concurrent arrival at $a^{f,i}$. If $A(t)$ and $A^{*} (t) $ are defined on $[0, t)$, this virtual delay definition defines the virtual waiting time for a (possible virtual) arrival at time $t$.}
\begin{eqnarray}
D^{f,i} = D^{f}(a^{f,i})=\inf \{\tau: A^{*f}(a^{f,i}+\tau) \ge A^{f}(a^{f,i}) \} . &&
\end{eqnarray}

The min-plus convolution, denoted by $\otimes$, of functions $f(\cdot)$ and $g(\cdot)$, is defined as:
\begin{equation}\label{mip1}
f \otimes g (y) = \inf_{0 \le x \le y}\{f(x) + g(y-x)\}
\end{equation}
and it is easily verified $f \otimes g (y) = g \otimes f (y)$.

The maximum horizontal distance between functions $\alpha(\cdot)$ and $\beta(\cdot)$, denoted by $h(\alpha, \beta)$, is defined as
\begin{equation}\label{hd}
h(\alpha, \beta) = \sup_{s \ge 0}\{ \inf\{ \tau \ge 0: \alpha(s) \le \beta(s+\tau)\} \}.
\end{equation}

\section{Stochastic Network Calculus Basics} \label{sec-3}
In this section, some related stochastic network calculus models and existing results are introduced. 

\subsection{Models} 
In snetcal, stochastic arrival curve (SAC) and stochastic service curve (SSC) are the most fundamental models. While SAC is for traffic modeling, SSC is for server modeling. In the literature, there are several definition variations of SAC and SSC. In this paper, we adopt the following, to which the other variations may be mapped \cite{SNetCal}.

\begin{definition}\label{def-sac}
A flow is said to have a v.b.c (virtual backlog centric) stochastic arrival curve $\alpha(t)$ with bounding function $\bar{F}$, if its arrival process $A(t)$ satisfies, for any $t \ge 0$ \cite{Cruz96}\cite{LCN02}\cite{Jiang-comnet09},
\begin{equation} \label{sac-2}
P \{A(t) - A\otimes \alpha(t) > x \} \le \bar{F}(x)
\end{equation}
where $\alpha(t)$ is non-negative non-decreasing on $t$, and $\bar{F}(x)$ non-negative non-increasing on $x$.
\end{definition}

In Definition \ref{def-sac}, if $\bar{F}(0) = 0$, implying $\bar{F}(x) = 0$ for all $x \ge 0$ or in other words $A(t) \le A\otimes \alpha(t)$, $\alpha(t)$ is also called a (deterministic) arrival curve of the flow in the network calculus literature.

\begin{definition}\label{def-ssc}
A system is said to provide a stochastic service curve $\beta(t)$ with bounding function $\bar{G}$, if there holds, for all $t \ge 0$ \cite{Cruz96} \cite{SNetCal}, 
\begin{equation}
P\{A \otimes \beta(t) - A^{*}(t) >x \} \le \bar{G}(x)
\end{equation}
where $\beta(t)$ is non-negative non-decreasing on $t$, and $\bar{G}(x)$ non-negative non-increasing on $x$.
\end{definition}

In Definition \ref{def-ssc}, if $\bar{G}(0) = 0$, implying $\bar{G}(x) = 0$ for all $x \ge 0$ or in other words $A^{*}(t) \ge A \otimes \beta (t)$, $\beta(t)$ is also called a (deterministic) service curve of the system in the network calculus literature.

\subsection{Related Results}
This paper focuses on stochastic service curve and delay bound analysis. The following presents some related results.

For stochastic service curve analysis, due to the difficulties that will be discussed in the next section, very little is known for the general case where packet length distribution is taken into consideration, and available results mostly assume fluid system ignoring packetization effect or that all packets have the same length. 

In the SNC literature, the following result, called the {\em leftover service} property, has been widely used for finding the stochastic service curve charaterization of the service provided to a flow (e.g. see \cite{Ciucu06}\cite{Liu07}). 

\begin{proposition}\label{lm-los}
Consider a system with cross traffic. If the system provides a stochastic service curve $\beta(t)$ with bounding function $\bar{G}(x)$ and the crossing flow has a v.b.c. stochastic arrival curve $\alpha^c(t)$ with bounding function $\bar{F}^c(x)$, then the {\em leftover service} provided to the traversing flow has a stochastic service curve $\beta^f(t)=(\beta(t)-\alpha^c(t))^{+}$ with bounding function $\bar{G}^f(x) = \bar{F}^c \otimes \bar{G}(x)$. 
\end{proposition}

For delay bound analysis, the following result has been proved (e.g. see \cite{Cruz96}\cite{SNetCal}).

\begin{proposition}\label{lm-dbd}
If a system provides a stochastic service curve $\beta(t)$ with bounding function $\bar{G}$ to a flow $f$, which has v.b.c stochastic arrival curve $\alpha^f(t)$ with bounding function $\bar{F}^f$, then the flow has a delay bound as
\begin{equation}\label{sdb}
P \{ D^{f}(t) > h(\alpha^f +x, \beta) \} \le \bar{F}^f\otimes \bar{G}(x).
\end{equation}
\end{proposition}

When the system is shared by the traversing flow and the crossing flow, the following delay bound follows immediately from Proposition \ref{lm-los} and Proposition \ref{lm-dbd} \cite{SNetCal}.

\begin{proposition}\label{lm-los-db}
Under the same condition as Preposition \ref{lm-los}, if the traversing flow has a stochastic arrival curve $\alpha^f(t)$ with bounding function $\bar{F}^f(x)$, then the delay of the flow is bounded as
\begin{equation}\label{sdb-sim}
P \{ D^{f}(t) > h(\alpha^f +x, \beta^{f}) \}  \le \bar{F}^f \otimes \bar{F}^c\otimes \bar{G}(x).
\end{equation}
where $\beta^f(t) = (\beta(t) - \alpha^c(t))^{+}$.
\end{proposition}

\section{The Difficulties}\label{sec-4}

As highlighted in the previous section, stochastic service curve (SSC) is the most fundamental server model for snetcal. As reviewed there, if the SSC characterization of the service provided by the node to the flow and the v.b.c SAC characterization of the flow are known, a delay bound can be readily obtained from the existing snetcal results. In the literature, while many results (e.g. \cite{Jiang-note10}\cite{Jiang-comnet09}\cite{Kelly96}) may be exploited to find the v.b.c SAC characterization of a flow, there are very few for SSC analysis. 

The difficulties are inherent in the SSC definition. Suppose $S(t)$ is the service process provided by the node to the flow. The following equation, called the min-plus convolution queueing principle \cite{Jiang-valuetools09}, holds \cite{SNetCal}:
\begin{equation}\label{pri}
A^{*}(t) = \inf_{0 \le s \le t}\{A(s) + S(s, t) \}
\end{equation}
where $S(t)$ denotes the service process provided to the flow and $S(s,t) \equiv S(t) - S(s)$.

Essentially, SSC defines a way to characterize the service process $S(t)$. While the SSC definition allows the derivation of results useful for performance study of computer networks, finding the SSC characterization of a system is surprisingly challenging. Even for the simplest constant capacity single node system, the challenge already exists. 

\subsection{A Pitfall}\label{sec-4a}

When the node has constant capacity $C$ (in bps), the following equation has sometimes been {\em wrongly} believed in the network calculus literature (see, e.g., \cite{KumarMK04}), 
\begin{equation}\label{pitfall}
A^{*}(t) = \inf_{0 \le s \le t}\{A(s) + C \cdot (t-s) \}.
\end{equation}
Or in other words, it is {\em wrongly} believed, for the constant capacity node:
\begin{equation}\label{pitfall0}
S(s, t) = C \cdot (t-s). 
\end{equation}

To give a counterexample, let's consider a single packet flow input. The packet arrives at time $a^{f,1} = 1$ and has length $2$. It is then clear that $A(0)=0, A(1)=2, A(2)=2, A(3)=2$. Suppose $C=1$. Then, $d^{f, 1}=3$. Hence, $A^*(0)=0, A^*(1)=0, A^*(2)=0, A^*(3)=2$. However, from (\ref{pitfall}), the output would be $A^*(0)=0, A^*(1)=1, A^*(2)=2, A^*(3)=2$, which is wrong\footnote{Choosing to define $A(t)$ and $A^*(t)$ on $[0, t)$ does not correct the mistake.}. Table \ref{tb1} summarizes the comparison. 

\begin{table}[htb]
  \centering
\caption{A counterexample}
\label{tb1}
\begin{tabular}{|l|c|c|c|c|c|}
\hline
$t$ & 0 & 1 & 2 & 3 & 4\\ \hline
\hline
$A(t)$& 0 & 2 & 2 & 2 & 2\\ \hline
$A^*(t)$ actual & 0 & 0 & 0 & 2 & 2\\ \hline
$A^*(t)$ from (\ref{pitfall}) & 0 & 1 & 2 & 2 &2 \\ \hline
\end{tabular}%
\end{table}

\subsection{Difficulty in Finding SSC}

When fluid-flow is assumed, i.e. $l^{f, i} \to 0$ and  $l^{c, j} \to 0$  for all packets, it is easy to verify that the amount of service provided by the node during any backlog period with length $\tau$ is $C \cdot \tau$. Then, from the network calculus literature, it is known that the node provides a deterministic service curve $\beta(t) = C \cdot t$\nop{\footnote{Indeed, in this case, it can be proved that $A^{*}(t) = \inf_{0 \le s \le t}\{A(s) + C\cdot(t-s) \}$.}}. In addition, when cross traffic is present, the stochastic service curve provided to the traversing flow is readily derived from the network calculus {\em leftover service} property as shown by Proposition \ref{lm-los}. In addition, from the traversing flow's viewpoint, the crossing flow can be treated as a process that impairs the total service provided by the server. Then with the impairment process concept \cite{Sigcomm06}\cite{SNetCal}, the stochastic service curve characterization of the service available to the traversing flow can also be found.

However, when packetization effect is taken into account, finding stochastic service curves for the node becomes {\em surprisingly} challenging, even though it has constant capacity.

Indeed, the network calculus literature has shown that a constant server with capacity $C$ has a deterministic service curve $(C\cdot t - L^{max})^{+}$, where $L^{max}$ denotes the maximum packet length in the system. This follows from two fundamental results. (i) If there is a function that lower-bounds the amount of service provided to the input during any backlog period, then the function is a (deterministic) service curve of the server \cite{NetCal}. (ii) Within any backlog period of length $t$, the amount of service provided by a constant rate server with capacity $C$  is lower bounded by $(C\cdot t - L^{max})^{+}$ \cite{Jiang01}. 

Fundamentally, the following inequality can be proved\nop{\footnote{We remark that, while under certain assumptions, such as fluid-flow, it can be proved that $A^{*}(t) = \inf_{0 \le s \le t}\{A(s) + C\cdot(t-s) \}$, this quality generally does not hold and counter examples can be constructed.}} \cite{Jiang01}: 
\begin{equation}
A^{*}(t) \ge A(t)\otimes (C\cdot t - L^{max}(t))^{+}.
\end{equation}
where $L^{max}(t) \equiv \max\{l^1, l^2, l^{(t)} \}$ with $l^{(t)}$ denoting the length of the most recent packet that arrived before or at $t$, and 
 $L^{max}=\lim_{t \to \infty} L^{max}(t)$.

With simple manipulation based on the definition of $\otimes$, we obtain
\begin{equation}
A(t)\otimes (C\cdot t) - A^{*}(t)  \le L^{max}(t)
\end{equation}
which implies 
\begin{eqnarray}\label{eq-lmax}
P\{A(t)\otimes (C\cdot t) - A^{*}(t) > x\} &\le& P\{ L^{max}(t) >x\} \nonumber \\
&\equiv& \bar{F}^{L^{max}(t)}(x). 
\end{eqnarray}

Then, we can conclude that {\em the constant capacity node provides a stochastic service curve $C \cdot t$ with bounding function $\bar{F}^{L^{max}(t)}(x)$}.

Unfortunately, $L^{max}(t)$ is non-decreasing with $t$, implying that $\bar{F}^{L^{max}(t)}(x)$ may approach $1$ as $t$ grows \cite{CiucuS12}\footnote{An exception is when all packets have the same length $L$ or their lengths are upper-bounded by $L$. In this case, for any $t$, $\bar{F}^{L^{max}(t)}(x) \le 1$ for all $x \le L$; otherwise $\bar{F}^{L^{max}(t)}(x) = 0$ for all $x > L$. Under this case, the node provides a deterministic service curve $(C\cdot t - L)^{+}$, with which, further analysis similar to that under the fluid-flow case can be conducted.\nop{ However, in general, for a process like $\{l^1, l^2, \dots \}$, where the packet length is distributed between 0 and $\infty$, it can been shown that as $t \to \infty$, $\lim_{t \to \infty} P\{ L^{max}(t) >x\}$ can only take value one \cite{CiucuS12}.}}. Consequently, using  $\bar{F}^{L^{max}(t)}$ as a bounding function is meaningless.

The problem becomes even more challenging when there is cross traffic. First, in order to apply the leftover service property to obtain a stochastic service curve for the traversing flow, we need to know the stochastic service curve of the node. However, the above discussion implies that the stochastic service curve of the node is yet to be found. Second\nop{ if the compromised service curve would be used}, the packet length process is a mixture of the packet length process of the traversing flow and that of the crossing flow. This makes the determination of $\bar{F}^{L^{max}(t)}$ and consequently the stochastic service curve characterization of the node even more difficult. 

To tackle the time-growing $\bar{F}^{L^{max}(t)}$ problem, one may introduce a compromised service curve $(C - \theta) \cdot t$, for some $\theta \ge 0$, with a resultant bounding function related to $\int_{x}^{\infty}\bar{F^{l}}(y)dy$, by exploiting an approach used in SNC in dealing with maximal random processes \cite{SNetCal}. Recently, the effort in \cite{LiebeherrBC12} shows that, without compromising the service curve expression $C \cdot t$, a bounding function, which is also related to $$\int_{x}^{\infty}\bar{F}^{l}(y)dy$$ can be found, when the packet length process is stationary and satisfies some conditions.

Note that intuitively, a packetized system can be treated as the concatenation of a fluid system followed by a packetizer \cite{Chang00}\cite{NetCal} \footnote{A packetizer is an element gathering all bits in a packet, which delivers the entire packet with no delay until and immediately after receiving the last bit of the packet.}. Since the fluid system provides deterministic service curve $C \cdot t$ as discussed above, the stochastic behavior of the node is hence determined by the packet length distribution. Based on this, we boldly conjecture that {\em the constant capacity node provides a stochastic service curve $C \cdot t$ with bounding function simply as $\bar{F}^{l}$}. However, while the intuition is perhaps straightforward, proving the validity of the conjecture is far from direct as to be shown in the next section.

\subsection{Difficulty in Making Use of Independence Information}
Besides the difficulty in finding the SSC characterization of the node, it is even more difficult to make use of potential independence information in the analysis. This is due to that, the service process $S(t)$ and the arrival process $A(t)$ are inherently dependent, implied by (\ref{pri}). More specifically, both $S(t)$ and $A(t)$ are functions of the lengths of packets that are counted in. For a simple example, suppose flow $f$ only has one packet $p^{f,1}$ whose length $l^{f, 1}$ is a random variable. It is clear that $A(t) = l^{f, 1}$ and also $S(t) = l^{f, 1}$ for any $t \ge d^{f, 1}$, which indicates strong dependence between $A(t)$ and $S(t)$. 

The inherent dependence between $A(t)$ and $S(t)$ makes it difficult to make use of potential independence in the analysis. Particularly, when there is cross traffic present, even though the traversing flow may be independent of the crossing flow, this independence information cannot be exploited when applying the existing snetcal results as reviewed in Section \ref{sec-3}. This is due to that, the stochastic service curve characterization of the service process $S(t)$ provided by the node is dependent on the packet lengths of both flows. A consequence is that, with the SSC decided from Proposition \ref{lm-los}, it is not possible to make use of the dependence information to improve the (independence-information-unaware) delay bound in Proposition \ref{lm-los-db}.

To this point, we would like to remark that, if all packets have the same length\footnote{The fluid-flow is a special case with infinitely small packet length.} or their lengths are upper-bounded, the service process of the node has a deterministic service curve\nop{ as discussed in the previous subsection}. For this case, independent case analysis may be conducted by following the approaches proposed in \cite{Sigcomm06} and \cite{Fidler06b}. 

However, for the more general case, even though it is intuitive that the independence information of the two flows should allow improving the analysis, how specifically to make use of this independence information in the analysis remains to be addressed.

\section{Stochastic Service Curve and Delay Bounds: The Single Flow Case} \label{sec-5}

In this section, we first prove the stochastic service curve as suggested by the conjecture. Then, delay bounds are derived for the traversing flow. To deal with the difficulty in finding the SSC, a novel approach is introduced. 

\subsection{Stochastic Service Curve}

We now present the approach to tackling the difficulty. In this approach, we relate the service provided by a (not-necessarily constant rate) system to a flow $f$ to \nop{a virtual constant rate system whose capacity is $R$ (in bps) and its behavior is described by }a virtual time function\nop{. The virtual time function is} defined as 
\begin{eqnarray}
V^{f,i}(R) &=& \max \{a^{f,i}, V^{f,i-1} \} + \frac{l^{f,i}}{R} \label{eq-vtf}
\end{eqnarray}
iteratively for $i=1, 2 \dots$, with $V^{f,0}=0$, where $R$ is a constant rate parameter. 

Applying iteratively to its right hand side, (\ref{eq-vtf}) becomes 
\begin{eqnarray}
V^{f,i}(R) &=& \max_{1 \le j \le i} \{a^{f,j}+ \frac{\sum_{k=j}^{i}l^{f,k}}{R} \}.  \label{eq-vtfa}
\end{eqnarray}

The following result is crucial, which establishes a link between the stochastic service curve model and the virtual time function. For deterministic network calculus, a similar relationship can be found in \cite{Jiang03} (Lemma 2). The proof is long and is included in the appendix. 

\begin{lemma}\label{lm-relation}
Consider a flow $f$ served by a system. For any time $t >0$ and $R >0$, the following relationship holds (for any sample path of the system):
\begin{equation}\label{4ssd-1}
A^{f}\otimes\beta(t) - A^{f*}(t) \le R \cdot [d^{f,i(t)} - V^{f,i(t)}(R)] + l^{f, i(t)}
\end{equation}
where $\beta(t) = R \cdot t$, $i(t) = \min \{k: d^{f, k} \ge t \}$\footnote{Intuitively, if at time $t$, there is a packet under service from the flow, then $p^{f,i(t)}$ is this packet; otherwise, $p^{f,i(t)}$ is the first packet from this flow, which receives service after $t$.}, and $l^{f, i(t)}$ the length of packet $p^{f,i(t)}$. 
\end{lemma}

For the considered single node system with single input flow $f$, consider any packet $p^{f,i}$. There are two cases. One case is that when $p^{f,i}$ arrives, the system is idle, which is $a^{f,i} > d^{f, i-1}$. Hence, $d^{f,i} = a^{f,i} + \frac{l^{f,i}}{C}$. Another case is that $p^{f,i}$ arrives, the system is busy, which is $a^{f,i} \le d^{f, i-1}$. Then, it has to wait until the previous packet $p^{f,i-1}$ has finished service. Hence, $d^{f,i} = d^{f,i-1} + \frac{l^{f,i}}{C}$. Combining both cases, we must have $d^{f,i} = \max \{a^{f,i}, d^{f,i-1} \} + \frac{l^{f,i}}{C}$. Comparing it with $V^{f,i}(C)$, the following result is proved.

\begin{lemma}\label{lm-ssdwc}
For the considered single node system with single input flow $f$, there holds, for any packet $p^{f,i}$ of the flow,
\begin{equation}\label{eq-ab01}
d^{f,i} = V^{f,i}(C).
\end{equation}
\end{lemma}

Applying (\ref{eq-ab01}) to Lemma \ref{lm-relation}, the following is obtained for $R=C$:
\begin{equation} \label{pkt-1}
A^{f}\otimes\beta(t) - A^{f*}(t) \le l^{f, i(t)}.
\end{equation}

It is worth highlighting that $i(t)$ is random and packet $p^{f,i(t)}$ may be different from one sample path to another sample path.
\nop{
In general, we have
\begin{equation}
P\{ A^{f}\otimes\beta(t) - A^{f*}(t) >x\} \le P\{\max \{l^{f, 1}, \dots, l^{f, i(t)}\}>x\}
\end{equation}
which is consistent with the discussion in Section \ref{sec-4}. 
}
However, if all packets, $l^{f, 1}, l^{f, 2}, \dots$, have identically distributed packet lengths with CCDF $\bar{F}^{l}(x)$, or more generally if their lengths have the same upper-bounded CCDF $\bar{F}^{l}(x)$, the following follows from (\ref{pkt-1}). 
\begin{equation}
P\{ A^{f}\otimes\beta(t) - A^{f*}(t) >x\} \le \bar{F}^{l}(x).
\end{equation}

Summarizing the above discussion, we have validated the conjecture. Formally, the following theorem has been proved:

\begin{theorem}\label{th-abssc}
Consider a work-conserving system with constant capacity $C$ serving a flow $f$. \nop{Suppose the packet length process is stationary with packet length distribution CCDF (or upper-bounded CCDF) $\bar{F}^{l}(x)$.} Suppose that all packets have length distributions that are identical with CCDF $\bar{F}^{l}(x)$ or whose CCDFs are all upper-bounded by $\bar{F}^{l}(x)$. Then, the system provides to the flow a stochastic service curve $\beta(t) = C \cdot t$ with bounding function $\bar{G}(x) = \bar{F}^{l}(x)$.
\end{theorem}

\subsection{Delay Bounds}

With Theorem \ref{th-abssc}, the following delay bound follows directly from Proposition \ref{lm-dbd}.

\begin{corollary}\label{th-ssdwc-1a}
Under the same condition as for Theorem \ref{th-abssc}, if the traversing flow has a v.b.c stochastic arrive curve $\alpha(t) = r^f \cdot t$ with bounding function $\bar{F}^f$ and $r^f \le C$, then for any packet $p^{f,i}$, its delay is bounded as:
\begin{equation} \label{ssdwc-1a}
P\{D^{f,i} > \tau \} \le \bar{F}^l \otimes \bar{F}^f (C \cdot \tau) 
\end{equation}
\end{corollary}

In Section \ref{sec-4}, we have discussed the inherent dependence between the arrival process and the service process. When it comes to the delay bound analysis, we would like to highlight that the inherent dependence is specifically seen between $A^{f}(t)$ and $l^{f, i(t)}$, since by their definitions, $l^{f, i(t)}$ may be counted in $A^{f}(t)$ . This partly explains why the min-plus convolution appears on the right hand side of (\ref{ssdwc-1a}) and in Proposition \ref{lm-dbd}, which assumes no knowledge of potential independence information. 

At this moment, it seems that nothing more than Corollary~\ref{th-abssc} could be done for delay bound analysis. In the following, we show that this is too pessimistic. Specifically, by exploiting the idea of the virtual time function, an improved delay bound is proved in the following theorem.

\begin{theorem}\label{th-ssdwc-2a}
Under the same condition as for Theorem \ref{th-abssc}, if the traversing flow has a v.b.c stochastic arrive curve $\alpha(t) = r^f \cdot t$ with bounding function $\bar{F}^f$ and $r^f \le C$, then for any packet $p^{f,i}$, its delay is bounded as (a.s.):
\begin{equation}\label{ssdwc-2a}
P\{D^{f,i} > \tau \} \le \bar{F}^f (C \cdot \tau).
\end{equation}
\end{theorem}

\begin{proof}
Consider any sample path of the system. By the definition of $D^{f,i}$ and with Lemma \ref{lm-ssdwc}, we have
\begin{eqnarray}
&& D^{f,i} = d^{f,i} - a^{f,i} = V^{f,i}(C) - a^{f,i} \nonumber \\
&=& \frac{1}{C} \max_{1 \le j \le i} \{\sum_{k=j}^{i}l^{f,k} - C (a^{f,i}-a^{f,j}) \}  \nonumber \\
&\le& \frac{1}{C} \max_{0 \le j \le i} \{A^{f}(a^{f,j}-\epsilon, a^{f,i}) - C(a^{f,i} - a^{f,j} + \epsilon)\}  + \epsilon \nonumber \\ && \label{m-step-1}\\ 
&\le& \frac{1}{C} \sup_{0 \le s \le t} \{A^{f}(s, t) - r^{f} (t - s)\} + \epsilon 
\end{eqnarray}
where $t=a^{f,i}$, $\epsilon \to 0$ and $r^{f} \le C$. In step (\ref{m-step-1}), $\epsilon \to 0$ is introduced such that $A^{f}(a^{f,j}-\epsilon, a^{f,i})$ includes all arrivals in $[a^{f,j}, a^{f,i}]$. 

Since the traversing flow has a v.b.c stochastic arrive curve\nop{ $\alpha(t) = r^f \cdot t$ with bounding function $\bar{F}^f$}, we have by definition:
$$
P \{ \sup_{0 \le s \le t} \{A^{f}(s, t) - r^{f} (t - s)\} >x \} \le \bar{F}^f.
$$
Since this bounding function holds for all sample paths, (\ref{ssdwc-2a}) is then obtained. 
\end{proof}

\vspace{3mm}

At a first glance, the delay bound in Theorem \ref{th-ssdwc-2a} may seem to be surprising, since the packetization effect is not directly seen from (\ref{ssdwc-2a}). However, an alert reader may have noticed that it is indeed consistent with a result in the deterministic network calculus literature, which states that in delay bound analysis, {\em the last packetizer on the path of the flow may be ignored} \cite{Chang00}\cite{NetCal}. Theorem \ref{th-ssdwc-2a} proves this property in the context of stochastic network calculus for the single node case. 

{\bf Remark:} An implication of Theorem \ref{th-ssdwc-2a} is that, when delay bound analysis is performed, the node may be treated as if it would provide a deterministic service curve $C \cdot t$ and then Corollary \ref{th-ssdwc-1a} becomes the same as Theorem \ref{th-ssdwc-2a}. 

\section{Stochastic Service Curves and Delay Bounds: The Case with Cross Traffic} \label{sec-6}

In this section, we consider the case where the traversing flow shares service of the node with the crossing flow. Specifically, we find a stochastic service curve for the node and two SSCs for the traversing flow, followed by deriving delay bounds for the traversing flow. 

\subsection{Stochastic Service Curves}

\subsubsection{A direct result}

Let us treat the traversing flow and the crossing flow as an aggregate flow. For packets of the aggregate flow $g$, {\em which takes the packet order as that on the output link}, the following relation can be easily verified:
\begin{equation}\label{eq-ab1}
d^{g,j} = \max\{a^{g,j}, d^{g,j-1}\} + \frac{l^{g,j}}{C}.
\end{equation}
Comparing (\ref{eq-ab1}) with (\ref{eq-vtf}), it is clear that for the aggregate, $d^{g,j} = V^{g,j}(C)$. 

Note that in presenting (\ref{eq-ab1}), we do not make any assumption on the scheduling algorithm between the two flows, and (\ref{eq-ab1}) is only concerned about the aggregate. We call (\ref{eq-ab1}) the ``aggregate behavior'' of the node, which is in consistence with the definition of the aggregate per-hop behavior \cite{Jiang-ton06} under the Differentiated Services (DiffServ)\nop{ Internet QoS} architecture \cite{DiffServ}. 

With (\ref{eq-ab1}) and following the same proof of Theorem \ref{th-abssc}, the following result can be verified.

\begin{lemma}\label{lm-abssc}
Consider a work-conserving system with constant capacity $C$, shared by a traversing flow $f$ and a crossing flow $f^{c}$. \nop{Suppose all packets in the system have a stationary distribution with CCDF (or upper-bounded CCDF) $\bar{F}^{l^g}(x)$.} Suppose that all packets have length distributions that are identical with CCDF $\bar{F}^{l^g}(x)$ or whose CCDFs are all upper-bounded by $\bar{F}^{l^g}(x)$. Then, the system provides to the aggregate of the two flows an ``aggregate behavior'' stochastic service curve $C \cdot t$ with bounding function $\bar{F}^{l^g}(x)$.
\end{lemma}

Recall that we are interested in the traversing flow. With Lemma \ref{lm-abssc} and the leftover property Proposition \ref{lm-los}, the following stochastic service curve to the traversing flow is obtained\footnote{Strictly speaking, instead of directly applying Proposition \ref{lm-los}, a separate proof is needed.}:

\begin{theorem}\label{th-sscct}
Consider the same system as in Lemma \ref{lm-abssc}. If the the crossing flow has a v.b.c. stochastic arrival curve $\alpha^c(t)= r^{c} \cdot t$, ($r^c < C$), with bounding function $\bar{F}^c(x)$, then the system provides to the traversing flow a stochastic service curve $\beta(t) = (C - r^{c})\cdot t$ with bounding function $\bar{G}(x)$ as 
\begin{eqnarray}
\bar{G}(x) &=& \bar{F}^c \otimes \bar{F}^{l^g}(x) \label{ssc-1}
\end{eqnarray}
\end{theorem}

Note that in Theorem \ref{th-sscct}, the resulting bounding function is $\bar{F}^c \otimes \bar{F}^{l^g}(x)$, which assumes no knowledge of potential independence information, even though the crossing flow may be independent of the traversing flow. This is due to that ${F}^{l^g}$ is the length distribution of all packets, which include packets of the crossing flow, and hence ${F}^{l^g}$ is inherently coupled with the traffic arrival process of the crossing flow. 


\subsubsection{An improved result}

While Theorem \ref{th-sscct} is an improvement over those that are based on (\ref{eq-lmax}), it may be difficult to find $\bar{F}^{l^g}$ of the aggregate particularly when the traversing flow and the crossing flow have different packet length distributions. 

The following theorem proves another stochastic service curve for the traversing flow, where there is no need to find $\bar{F}^{l^g}$, relieving the difficulty. In addition, if the two flows are independent, this independence information is made use of.

\begin{theorem}\label{th-sscwc}
Consider a work-conserving system with constant capacity $C$, shared by a traversing flow $f$ and a crossing flow $f^{c}$. 
Suppose the crossing flow has a v.b.c stochastic arrival curve $\alpha^{c}(t) = r^c \cdot t$, ($r^c < C$), with bounding function $\bar{F}^c$, and suppose all packets of the traversing flow have length distributions that are identical with CCDF $\bar{F}^{l}(x)$ or whose CCDFs are all upper-bounded by $\bar{F}^{l}(x)$. \nop{suppose all packets of the traversing flow have a stationary distribution whose CCDF is (or their length distributions are upper-bounded by) $\bar{F}_{l}(x)$.} Then, the node provides to the traversing flow a stochastic service curve $\beta(t)= (C - r^{c})$ with bounding function $\bar{G}(x)$ as 
\begin{eqnarray}
\bar{G}(x) &=& \bar{F}^c \otimes \bar{F}^l (x) \label{ssc-2}
\end{eqnarray}
and if the two flows are independent,
\begin{eqnarray}
\bar{G}(x) &=& 1- F^{c} * F^{l} (x) \label{ssc-2i}
\end{eqnarray}
where, $F^{l} \equiv 1- \bar{F}^l$, $F^{c} \equiv 1- \bar{F}^c$, and $F_1 * F_2 (x) \equiv \int_{0}^{x}{F}_1(x-y)d {F}_2(y)$. 
\end{theorem}

To prove Theorem \ref{th-sscwc}, Lemma \ref{lm-sscwc} and Lemma \ref{lm-sscwc-2} below are crucial, with which, Theorem \ref{th-sscwc} is easily verified. 
\begin{lemma}\label{lm-sscwc}
For the considered single node system with cross traffic, there holds, for any packet $p^{f,i}$ of the traversing flow, 
\begin{equation} \label{eq-10}
d^{f,i} \le V^{f,i}(C-r^c) + \frac{\sup_{0\le s \le d^{f,i}} \{A^{c}(s, d^{f,i}) - r^c(d^{f,i}-s) \}}{C-r^c}
\end{equation}
for any $C > r^c \ge 0$.
\end{lemma}

\begin{proof}
As for (\ref{eq-ab1}), let us consider the aggregate flow $g$\nop{ on the output link}. Since no specific scheduling between the two flows has been assumed, a packet, which appears earlier on the output link in the aggregate flow, may actually arrive to the node later than another packet that appears later on the output link. In other words, we may not have $a^{g,j} \ge a^{g, j-1}$. 

For any packet $p^{f,i}$, suppose its corresponding packet in the aggregate flow $g$ is $p^{g,j}$. Particularly, we suppose the departure time of $p^{f,i}$, i.e. $d^{f,i}=d^{g,j}$, is within the busy period that starts at $t^0$. Note that such a busy period always exists, since in the worst case, the period is only the service time period of $p^{f,i}$ and in this case, $t^0=a^{g,j}$. 

Since the node is work-conserving with constant service rate $C$ and it is busy with serving between $t^0$ and $d^{g,j}$, there holds:
\begin{equation}\label{eq-2}
d^{g,j} = t^0 + \frac{\sum_{k=j_0}^{j} l^{g,k}}{C},
\end{equation}
where $p^{g,j_0}$ denotes the packet whose arrival starts the busy period. 

Among packets $p^{g,j_0}, \dots, p^{g,j}$, some belong to the traversing flow and the rest the crossing flow. Let $p^{f,i_0}$ denote the first packet from the traversing flow served in the busy period. There holds $a^{f,i_0} \ge t^0$. 

Equation (\ref{eq-2}) can be re-written as:
\begin{equation}\label{eq-3a}
d^{f,i} \le t^0 + \frac{\sum_{k=i_0}^{i} l^{f,k}}{C} + \frac{A^{*c}(t^0, d^{f,i})}{C},
\end{equation}
where, by definition, $A^{*c}(t^0, d^{f,i})$ represents the total length (in bits) of packets from the crossing flow served in $(t^0, d^{f,i}]$\footnote{Note that $t^0$ starts the busy period and hence no packet finishes service at  $t^0$. This implies $A^{*c}(t^0, d^{f,i})$ indeed represents the total length (in bits) of packets from the crossing flow served in $[t^0, d^{f,i}]$.}.

Since the busy period starts at $t^0$, this implies that immediately before $t^0$, the node is idle. In other words, all packets, which arrived before $t^0$, have been served by $t^0$. So, we have $A^{*c}(t^0) = A^{c}(t^0)$. In addition, crossing flow packets, which are served before $d^{f,i}$, must have arrived by $d^{f,i}$. So, we have $A^{*c}(d^{f,i}) \le A^{c}(d^{f,i})$. Combing both, we obtain:
\begin{equation}
A^{*c}(t^0, d^{f,i}) \le A^{c}(t^0, d^{f,i})
\end{equation}
which, when applied to (\ref{eq-3a}), results in
\begin{equation}\label{eq-5a}
d^{f,i} \le t^0 + \frac{\sum_{k=i_0}^{i} l^{f,k}}{C} + \frac{A^{c}(t^0, d^{f,i})}{C}.
\end{equation}

With (\ref{eq-5a}), we obtain, for any $C > r^c \ge 0$,
\begin{eqnarray}
d^{f,i} &\le& t^0 + \frac{\sum_{k=i_0}^{i} l^{f,k}}{C} + \frac{A^{c}(t^0, d^{f,i}) - r^c (d^{f,i}-t^0)}{C} \nonumber\\
&& + \frac{r^c (d^{f,i}-t^0)}{C} \nonumber
\end{eqnarray}
Further with simple manipulation, we obtain
\begin{equation}\label{eq-6a}
d^{f,i} \le t^0 + \frac{\sum_{k=i_0}^{i} l^{f,k}}{C-r^c} + \frac{A^{c}(t^0, d^{f,i}) - r^c (d^{f,i}-t^0)}{C-r^c} 
\end{equation}

Recall the virtual time function (\ref{eq-vtf}), it is ease to verify that, for the considered packet $p^{f,i}$, we have
\begin{eqnarray}
V^{f,i}(C-r^c) &\ge& a^{f,i_0} + \frac{\sum_{k=i_0}^{i} l^{f,k}}{C-r^c} \ge t_0 + \frac{\sum_{k=i_0}^{i} l^{f,k}}{C-r^c} \nonumber\\ \label{eq-8a}
\end{eqnarray}
In addition, there holds
\begin{eqnarray}
&& A^{c}(t^0, d^{f,i}) - r^c(d^{f,i}-t^0) \nonumber \\
&\le& \sup_{0\le s \le d^{f,i}} \{A^{c}(s, d^{f,i}) - r^c (d^{f,i}-s) \} \label{eq-9a}
\end{eqnarray}

Applying (\ref{eq-8a}) and (\ref{eq-9a}) to (\ref{eq-6a}), we obtain (\ref{eq-10}) and the lemma is proved.
\end{proof}

We remark that when there is no cross traffic, i.e. $A^c(t)=0$, then letting $r^c=0$, Lemma \ref{lm-sscwc} is reduced to Lemma \ref{lm-ssdwc} as expected.

Note that Lemma \ref{lm-relation} provides a general relationship, with which, by letting $R= C-r^c$ in it, we obtain
\begin{eqnarray}
&& A^{f} \otimes \beta (t) - A^{f*}(t) \nonumber \\
&\le& (C-r^c)[d^{f, i(t)} - V^{f, i(t)} (C-r^c)] + l^{f, i(t)}
\end{eqnarray}
Applying Lemma \ref{lm-sscwc} to above immediately gives the following:

\begin{lemma}\label{lm-sscwc-2}
For the considered single node system with cross traffic, for any time $t$ and any sample path of the system, the following relationship holds for the traversing flow $f$,
\begin{eqnarray}\label{eq-11}
&& A^{f} \otimes \beta (t) - A^{f*}(t) \\
&\le& \sup_{0\le s \le d^{f,i(t)}} \{A^{c}(s, d^{f, i(t)}) - r^c \cdot (d^{f, i(t)}-s) \} + l^{f, i(t)} \nonumber
\end{eqnarray}
where $\beta(t) = (C- r^c) \cdot t$, $i(t) = \min \{k: d^{f, k} \ge t \}$\nop{ implying $p^{f,i(t)}$ to be the most recent packet in flow $f$ which departs the system not before $t$}, and $l^{f, i(t)}$ is the length of packet $p^{f, i(t)}$.
\end{lemma}

Finally, since the crossing flow has a v.b.c stochastic arrival curve $r^c t$ with bounding function $\bar{F}^c$, and all packet lengths have identical (or the same upper-bounded) CCDF $\bar{F}_l$, Theorem \ref{th-sscwc} is proved by applying these conditions to Lemma \ref{lm-sscwc-2}.

\vspace{3mm}

It is worth highlighting that while (\ref{ssc-2}) looks similar to (\ref{ssc-1}), there is a fundamental difference between them. Specifically, the packet length distribution ${F}^{l^g}$ in (\ref{ssc-1}) is that of the aggregate flow, while in (\ref{ssc-2}), the packet length distribution ${F}^{l}$ is only of the traversing flow.

\subsection{Delay Bounds}

\subsubsection{Delay bounds from Theorems \ref{th-sscct} and \ref{th-sscwc}}

With Theorems \ref{th-sscct} and \ref{th-sscwc}, the following delay bounds are directly obtained from Preposition \ref{lm-dbd} respectively. 

\begin{corollary}\label{th-sscwc-db0}
Under the same condition as for Theorem \ref{th-sscct}, if the traversing flow has a v.b.c stochastic arrival curve $\alpha(t) = r^{f} \cdot t$ with bounding function $\bar{F}^f$, and $r^{f} < C-r^c$, then for any packet $p^{f,i}$, it has a delay bound as:
\begin{equation} \label{2-db1}
P\{D^{f,i} > \tau\} \le \bar{F}^c \otimes \bar{F}^{l^g} \otimes \bar{F}^f((C-r^c)\tau).
\end{equation}
\end{corollary}

\begin{corollary}\label{th-sscwc-db1}
Under the same condition as for Theorem \ref{th-sscwc}, if the traversing flow has a v.b.c stochastic arrival curve $\alpha(t) = r^{f} \cdot t$ with bounding function $\bar{F}^f$, and $\rho^{f} < C-r^c$, then for any packet $p^{f,i}$, it has a delay bound as:\\
(i) if the two flows may be dependent, 
\begin{equation} \label{2-db2}
P\{D^{f,i} > \tau\} \le \bar{F}^c \otimes \bar{F}^l \otimes \bar{F}^f((C-r^c)\tau); 
\end{equation}
(ii) if the two flows are independent,
\begin{equation}\label{2-db3}
P\{D^{f,i} > \tau\} \le (1- F^{c} * F^{l}) \otimes \bar{F}^f((C-r^c)\tau).
\end{equation}
\nop{
\begin{equation} \label{sscwc-db1}
P\{D^{f,i} > \tau\} \le \bar{G} \otimes \bar{F}^f((C-r^c)\tau)
\end{equation}
where if the two flows may be dependent, 
\begin{eqnarray}
\bar{G}(x) &=& \bar{F}^c \otimes \bar{F}^l (x) \nonumber 
\end{eqnarray}
and if the two flows are independent,
\begin{eqnarray}
\bar{G}(x) &=& 1- F^{c} * F^{l} (x). \nonumber 
\end{eqnarray}
}
\end{corollary}

It is worth highlighting that in obtaining the bounding function in Theorem \ref{th-sscwc}, we have relied on the right hand side of (\ref{eq-11}), which and $A(t)$ are inherently dependent due to $l^{f, i(t)}$. This explains why in (\ref{2-db3}), the independence information cannot be further made use of.

\subsubsection{An improved delay bound}

In the following, an improved delay bound is presented. 

\begin{theorem}\label{th-sscwc-db2}
Suppose the traversing flow has a v.b.c stochastic arrival curve $\alpha(t) = r^f \cdot t$ with bounding function $\bar{F}^f$ and the crossing flow has a v.b.c stochastic arrival curve $\alpha(t) = r^c \cdot t$ with bounding function $\bar{F}^c$. If $r^f + r^c < C$, then for any packet $p^{f,i}$ of the traversing flow, its delay is bounded as (a.s.) \\
(i) if the two flows may be dependent,
\begin{equation}\label{2-db4}
P\{D^{f,i} > \tau\} \le \bar{F}^c \otimes \bar{F}^f ((C-r^c)\tau); 
\end{equation}
(ii) if the two flows are independent, 
\begin{equation}\label{2-db5}
P\{D^{f,i} > \tau\} \le 1-{F}^c * {F}^f ((C-r^c)\tau).
\end{equation}
\end{theorem}

\begin{proof}
Consider any sample path. With Lemma \ref{lm-sscwc} particularly (\ref{eq-10}), we obtain, for any packet $p^{f,i}$, 
\begin{eqnarray}
&& D^{f,i} = d^{f,i} - a^{f,i} \nonumber\\
&\le& V^{f,i}(C-r^c) - a^{f,i} \nonumber\\
&& + \frac{\sup_{0\le s \le d^{f,i}} \{A^{c}(s, d^{f,i}) - r^c\cdot(d^{f,i}-s) \}}{C-r^c} \nonumber \\
&=&  \max_{1 \le j \le i} \{a^{f,j}+ \frac{\sum_{k=j}^{i}l^{f,k}}{C-r^c} \} - a^{f,i} \nonumber \\
&& + \frac{\sup_{0\le s \le d^{f,i}} \{A^{c}(s, d^{f,i}) - r^c\cdot(d^{f,i}-s) \}}{C-r^c} 
\end{eqnarray}
To ease the presentation, we move $(C-r^c)$ to the left and get
\begin{eqnarray}
&& D^{f,i} \cdot (C-r^c) \nonumber \\
&\le& \max_{1 \le j \le i} \{\sum_{k=j}^{i}l^{f,k} - (C-r^c) (a^{f,i}-a^{f,j}) \}  \nonumber \\
&& + \sup_{0\le s \le d^{f,i}} \{A^{c}(s, d^{f,i}) - r^c\cdot(d^{f,i}-s) \} \nonumber \\
&\le& \max_{0 \le j \le i} \{A^{f}(a^{f,j}-\epsilon, a^{f,i}) - (C-r^c) (a^{f,i} - a^{f,j} + \epsilon)\} \nonumber \\ 
&& + \sup_{0\le s \le d^{f,i}} \{A^{c}(s, d^{f,i}) - r^c\cdot(d^{f,i}-s) \}  + \epsilon \label{a-1} \\
&\le& \sup_{0 \le s \le a^{f,i}} \{A^{f}(s, a^{f,i}) - r^{f} (a^{f,i} - s)\} \nonumber \\
&& + \sup_{0\le s \le d^{f,i}} \{A^{c}(s, d^{f,i}) - r^c\cdot(d^{f,i}-s) \}  + \epsilon  \label{a-2}
\end{eqnarray}
where $\epsilon \to 0$. 

Note that, given $a^{f,i}$ as implied by the delay definition, the first two terms on the right hand side of (\ref{a-2}) are independent.
This independence is more easily seen by expending them as 
\begin{eqnarray}
\sup_{0 \le s \le a^{f,i}} \{A^{f}(s, a^{f,i}) - r^{f} (a^{f,i} - s)\} +  && \nonumber \\
\max\left\{ \sup_{0\le s \le a^{f,i}} \{A^{c}(s, a^{f,i}) - r^c\cdot(a^{f,i}-s) \} \right. &&  \nonumber \\
\qquad + A^{c}(a^{f,i}, d^{f,i}) - r^c\cdot(d^{f,i}-a^{f,i}), &&  \nonumber \\
\left. \sup_{a^{f,i} < s \le d^{f,i}} \{A^{c}(s, d^{f,i}) - r^c\cdot(d^{f,i}-s) \} \right\}  && \label{a-2a}
\end{eqnarray}
where the first term is determined only by the time period $[0, a^{f,i}]$ and the arrivals of the traversing flow in this period, while the second term is determined by the same period $[0, a^{f,i}]$ and another later non-overlapping period $(a^{f,i}, d^{f,i}]$ and the arrivals of the crossing flow in these periods. Since for the same period $[0, a^{f,i}]$, the two arrival processes are independent and for the second period, the first term is not affected, the independence is hence concluded. Consequently, the theorem follows from (\ref{a-2}). 
\end{proof}

\subsubsection{A further improved delay bound}
In obtaining the improved delay bounds in Theorem \ref{th-sscwc-db2}, we made no assumption on the arrival process of the traversing flow or that of the crossing flow. If, however, these processes satisfy some assumptions, a further improved delay bound can be obtained. 

Specifically, if $A^{f} (t)$ and $A^{c} (t)$ are independent and they have independent stationary increments, a further improved delay bound can be obtained.

\begin{theorem}\label{th-sscwc-db3}
Suppose that the traversing flow $A^{f} (t)$ and the crossing flow $A^{c} (t)$ are independent and they have independent stationary increments. Assume $M^f(1) \equiv {E[e^{\theta A^f(1)}]}$ and $M^c(1) \equiv {E[e^{\theta A^c(1)}]}$ exist  for small $\theta > 0$ and $E[e^{\theta (A^f(1)+A^c(1)- C)}] \le 1$. Then, for any packet $p^{f,i}$ of the traversing flow, its delay is bounded as 
\begin{equation}\label{2-db6}
P\{D^{f,i} > \tau\} \le e^ {- \theta (C-r^c)\tau}.
\end{equation}
for any $\theta \ge 0$ and any $r^c$ such that $E[e^{\theta (A^c(1)- r^c)}] \le 1$. 
\end{theorem}

\begin{proof}
Our starting point is (\ref{eq-6a}), which is reproduced here:
\begin{equation}\label{eq-6a-rep}
d^{f,i} \le t^0 + \frac{\sum_{k=i_0}^{i} l^{f,k}}{C-r^c} + \frac{A^{c}(t^0, d^{f,i}) - r^c (d^{f,i}-t^0)}{C-r^c} 
\end{equation}
with which, the following is easily verified
\begin{eqnarray}
&& (C-r^c) \cdot (d^{f,i}-a^{f,i}) \nonumber \\
&\le& A^f(t^0, a^{f,i}) - (C-r^c) \cdot (a^{f,i}-t^0) \nonumber \\
&& + A^{c}(t^0, d^{f,i}) - r^c (d^{f,i}-t^0) \label{fidb-a1} \\
&=&  A^f(t^0, a^{f,i}) + A^{c}(t^0, a^{f,i}) - C \cdot (a^{f,i}-t^0) \nonumber \\
&& + A^{c}(a^{f,i}, d^{f,i}) - r^c (d^{f,i}-a^{f,i}) \nonumber \\
&\le & \sup_{0 \le s \le a^{f,i}} \{ A^f(s, a^{f,i}) + A^{c}(s, a^{f,i}) - C \cdot (a^{f,i}-s) \} \nonumber \\
&& + A^{c}(a^{f,i}, d^{f,i}) - r^c (d^{f,i}-a^{f,i}) \label{fidb-a2} \\
&= & \sup_{0 \le s \le a^{f,i}} \left \{ A^f(s, a^{f,i}) + A^{c}(s, a^{f,i}) - C \cdot (a^{f,i}-s) \right. \nonumber \\
&& \left. + A^{c}(a^{f,i}, d^{f,i}) - r^c (d^{f,i}-a^{f,i}) \right \} \label{fidb-a3}
\end{eqnarray}
where in step (\ref{fidb-a1}) we have used the fact that $\sum_{k=i_0}^{i} l^{f,k} \le A^f(t^0, a^{f,i})$. 

It is worth highlighting that, the two terms in (\ref{fidb-a2}) are independent, since the second term is determined by a period that is non-overlapping with the period involved in the first term, and the process $A^c(t)$ has independent increments. Also due to this, in step (\ref{fidb-a3}), we have intentionally moved the second term inside $\sup \{ \dot \}$. 

For ease of exposition, we let 
$$
Z = A^{c}(a^{f,i}, d^{f,i}) - r^c (d^{f,i}-a^{f,i})
$$
for which, it is easily verified that, $E[e^{\theta Z}|d^{f,i}] = (E[e^{\theta (A^c(1) - r^c)}])^{d^{f,i}-a^{f,i}}\le 1$ for $\forall d^{f,i}$ and hence $E[e^{\theta Z}] \le 1$, under the given assumptions. 

Then, for any $\theta \ge 0$, there holds, 
\begin{eqnarray}
&& P \{ (C-r^c) D^{f,i}  > x \} \nonumber \\
&=& P \{ e^{\theta (C-r^c) ( d^{f,i} - a^{f,i}) } > e^{\theta x} \}  \nonumber \\
&\le& P \{ e^{ \sup_{0 \le s \le a^{f,i}} \{ \theta[ A^f(s, a^{f,i}) + A^{c}(s, a^{f,i}) - C \cdot (a^{f,i}-s)] \} } \nonumber \\
&& \cdot e^{ \theta[ A^{c}(a^{f,i}, d^{f,i}) - r^c (d^{f,i}-a^{f,i}) ]} > e^{\theta x} \} \nonumber \\
&=& P \{\sup_{0 \le s \le a^{f,i}} e^{ \theta(A^{f}(s) + A^{c}(s) - C \cdot s)} \cdot e^{\theta Z } > e^{\theta x} \} \label{fidb-a4} \\
&\le& \frac{E[e^{ \theta(A^{f}(1) + A^{c}(1) - C)} e^{\theta Z }]}{e^{\theta x}} \label{fidb-a5} \\
&=& E[e^{ \theta(A^{f}(1) + A^{c}(1) - C)}] \cdot E[e^{\theta Z }] \cdot e^{-\theta x} \label{fidb-a5} \\
&\le & e^{-\theta x} \label{fidb-a6}
\end{eqnarray}
where step (\ref{fidb-a4}) is due to that both $A^f(t)$ and $A^c(t)$ are stationary processes, step (\ref{fidb-a5}) is from the Doob's maximal inequalities for sub-(super-)martingales, and step (\ref{fidb-a6}) is from the assumptions of the theorem.

Specifically, define $X(s) = e^{\theta(A^f(s) + A^c(s) - C \cdot s)} e^{\theta Z}$, $s = 0, 1, 2, \dots, a^{f,i}$. There holds, due to independent increments assumption,  
\begin{eqnarray}
&& E[X(s+1) | X(1), \dots, X(s)] \nonumber \\
&=& E[e^{\theta (A^f(s, s+1) + A^c(s,s+1) - C)}] X(s)\nonumber \\
&=& E[e^{\theta (A^f(1) + A^c(1) - C)}] X(s)\nonumber \\
&\le& X(s)
\end{eqnarray}
and hence $\{ X(s) \}$ forms a supermartingale. Then (\ref{fidb-a5}) is obtained from the Doob's maximal inequality for supermartingales, which has also been used in the snetcal literature \cite{Ciucu07b}\cite{Jiang-valuetools09}.
\end{proof}

\section{Examples} \label{sec-7}
To demonstrate the obtained results, examples are presented in this section. The focus is on the obtained delay bounds. Without loss of generality and for ease of expression, we normalize the capacity and take $C=1$. 

\subsection{Single Flow}

For the single flow case, consider the arrival process $A^f(t)$ governed by a compound Poisson process. In this process, packets arrive according to a Poisson process with intensity $\lambda$. Packet lengths are independent and identically distributed, following a negative exponential distribution with mean $\frac{1}{\mu}$. Specifically:
$$
A^f(t) = \sum_{n=1}^{N(t)} l^{f,i}
$$
where $N(t)$ is a Poisson process with arrival intensity $\lambda$, which is independent of the packet lengths, and $l^{f,1}, l^{f,2}, \dots$ are i.i.d. random variables with mean $\frac{1}{\mu}$.

For this compound Poisson process, it can be verified that it has a v.b.c. stochastic arrival curve \cite{Jiang-comnet09} \cite{Jiang-note10} $\alpha^f(t) = \frac{\lambda }{\mu - \theta} t$ with bounding function $\bar{F}^f(x) = e^{-\theta x}$ for $\forall \theta >0$ and $r^f \equiv \frac{\lambda }{\mu - \theta} \le 1$. Note that $r^f$ here is a function of $\theta$.

With Theorem \ref{th-ssdwc-2a}, under the condition that $r^f \le 1$, the tightest delay bound is obtained by taking $\theta=\mu - \lambda$, which is: 
\begin{equation}\label{ssdwc-mm1}
P\{D > \tau\} \le e^{-(\mu-\lambda) \tau}.
\end{equation}

It is worth highlighting that this single flow system may be considered\footnote{Note that in real computer networks the time is discrete. For this reason, we have also assumed discrete time at the beginning. Nevertheless, this paper does not specify the length of the time unit. Letting the unit time length $\to$ infinitely small, the system approaches time-continuous and all results in this paper still hold.} as an $M/M/1$ system with Poisson arrival rate $\lambda$ and exponential service time distribution with parameter $\mu$. 

Appealingly, the delay bound (\ref{ssdwc-mm1}) matches exactly with the delay\footnote{It is the delay in the system which matches the definition of $D^{f,i}$, while not the delay in queue.} distribution found from $M/M/1$ analysis.

\subsection{With Cross Traffic}
For the case with cross traffic, we suppose that priority scheduling is adopted, with the crossing flow at the high priority level.  

We assume the traversing flow and the crossing flow are independent of each other. For both, the arrival process is governed by a compound Poisson process. Similar to the single flow case, we consider that in each traffic arrival process, packets arrive according to a Poisson process with intensity $\lambda^f$ for the traversing flow (respectively $\lambda^c$ for the crossing flow). In addition, to ease later comparison, we assume all packets (of both flows) have the same i.i.d. length, following a negative exponential distribution with mean $1/\mu$. 

This system is equivalent to an $M/M/1/$priority system, for which, the classic queueing theory has exact result for the delay expectation of the low priority traffic. 

Note that, given the delay CCDF $P \{ D \ge \tau \}$, the average delay is obtained as \cite{Kleinrock75}
$$
E[D] = \int_{0}^{\infty} P \{ D \ge \tau \} d \tau.
$$
The above relationship between the delay expectation and the CCDF readily allows us to find upper bounds on delay expectation from the obtained delay bounds. Among the various delay bounds derived in the previous sections, (\ref{db-exp1}) and (\ref{db-exp2}) are the tightest and will be compared against the exact solution. 

\subsubsection{Delay expectation}
For the $M/M/1/$priority system, the classic queueing theory gives the following result:
\begin{equation}\label{ave-d}
E[D] = \frac{\rho}{\mu (1-\rho^c) (1- \rho)} +\frac{1}{\mu} = \frac{1}{\mu (1- \rho)}[1+\frac{\rho^c \rho}{1-\rho^c}]  
\end{equation}
where $E[D]$ denotes the delay expectation, $\rho^f \equiv \frac{\lambda^f}{\mu}$, $\rho^c \equiv \frac{\lambda^c}{\mu}$, and $\rho \equiv \rho^c + \rho^f$. 

\subsubsection{Bound on delay expectation, based on (\ref{2-db5})}

Again, for the two compound Poisson arrival processes, the traversing flow has a v.b.c. stochastic arrival curve $\alpha^f(t) = \frac{\lambda^f }{\mu - \theta^f} t$ with bounding function $\bar{F}^f(x) = e^{-\theta^f x}$ for $\forall \theta^f >0$; the traversing flow has a v.b.c. stochastic arrival curve $\alpha^c(t) = \frac{\lambda^c }{\mu - \theta^c} t$ with bounding function $\bar{F}(x) = e^{-\theta^c x}$ for $\forall \theta^c >0$. 

For ease of expression, letting $\theta^f=\theta^c \equiv \theta$, which may give a sub-tight bound, we obtain from Theorem~\ref{th-sscwc-db2} 
\begin{equation}\label{ssdwc-2aa}
P\{D > \tau \} \le (1 + \theta \cdot (1-r^c)y) e^{-\theta \cdot (1-r^c)\tau} 
\end{equation}
where $r^c = \frac{\lambda^c }{\mu - \theta}$ and $r^f = \frac{\lambda^f }{\mu - \theta}$, for any $\theta>0$, satisfying
$$
r^f + r^c \le 1 
$$
which further gives, by letting $\theta= \mu - \lambda^f -\lambda^c$
$$
P\{D > y\} \le [1 + \frac{\lambda^f (1 - \rho)}{\rho} y ]e^{- \frac{\lambda^f (1-\rho)}{\rho}y}
$$
and consequently, a bound on delay expectation is as:
\begin{equation}\label{db-exp1}
E[D] \le \frac{2}{\mu (1- \rho)}[1+\frac{\rho^c}{\rho^f}]. 
\end{equation}

\subsubsection{Bound on delay expectation, based on (\ref{2-db6})}

For the considered system, letting  $\theta= \mu - \lambda^f -\lambda_c$ and  $r^c = \frac{\lambda^c}{\lambda^f + \lambda^c}$, the following can be verified: (1) $E[e^{\theta (A^c(1)- r^c)}] = e^{\theta (\frac{\lambda^c}{\mu-\theta} - r^c)} =1$ and (2) $E[e^{\theta (A^f(1)+A^c(1)- C)}] = e^{\theta (\frac{\lambda^f}{\mu-\theta} +\frac{\lambda^c}{\mu-\theta} -1)} =1$. Then, from Theorem \ref{th-sscwc-db3}, the delay bound (\ref{2-db6}) becomes
\begin{equation}\label{ssdwc-2aa}
P\{D > \tau \} \le e^{-(\mu - \lambda^f -\lambda_c)(1-r^c)\tau}
\end{equation}
with which, the following bound on delay expectation is obtained:
\begin{equation}\label{db-exp2}
E[D] \le \frac{1}{\mu(1 - r^c)(1-\rho)} = \frac{1}{\mu (1- \rho)}[1+\frac{\rho^c}{\rho^f}]
\end{equation}
which is clearly better than (\ref{db-exp1}).

To give an overview of the bound (\ref{db-exp2}), Fig. \ref{fig-delay} is presented, where $x$-axis is the total load, $y$-axis is the share of cross traffic in the total load. In the figure, the bound is compared against the exact result (\ref{ave-d}), under different total loads ($\rho \in [0, 0.9]$), and different shares ($\frac{\rho^c}{\rho} \in [0, 0.9]$). The comparison shows that the bound (\ref{db-exp2}) is reasonably good. 

\begin{figure}[htb] 
  \centering
  \includegraphics[width=0.5\textwidth]{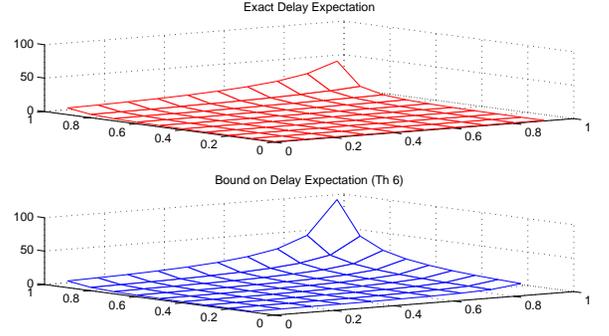}
  \caption{Comparison of bound (\ref{db-exp2}) ($\mu=1$)}
  \label{fig-delay}
\end{figure}

\section{Discussion and Related Work} \label{sec-8}


In deterministic network calculus, the delay bound derived from the Guaranteed Rate server model is better than that directly from the deterministic counterpart of (\ref{sdb}). To overcome this difference, an interesting property has been proved, which says, in deterministic delay bound analysis, {\em the last packetizer can be ignored} \cite{Chang00}\cite{NetCal}. For the considered single node case, this property implies that, for the concern of deterministic delay bound analysis, the constant capacity node could be treated as if it had a deterministic service curve $C\cdot t$ and hence Proposition \ref{lm-los-db} could be used directly. Results in this paper further imply that this property can also be extended to the stochastic network calculus context. Particularly, it is easily verified that, for the single flow case, delay bound (\ref{ssdwc-2a}) in Theorem \ref{th-ssdwc-2a} is better than delay bound (\ref{ssdwc-1a}) in Corollary  \ref{th-abssc} by ignoring the packetization effect $\bar{F}_l$. In addition, for the case with cross traffic, Corollary \ref{th-sscwc-db0} and Corollary \ref{th-sscwc-db1} will lead to Theorem \ref{th-sscwc-db3} by ignoring the packetization effect.


In the general sense of taking packetization effect into stochastic service curve and delay bound analysis, the work \cite{Burchard11} is most related. However, the obtained results in \cite{Burchard11} are mostly functions of $\int_{x}^{\infty}\bar{F^{l}}(y)dy$, while in our results, they are related directly to $\bar{F^{l}}$. In addition, how to make use of independence information to improve the obtained results is not investigated in \cite{Burchard11}. Moreover, \cite{Burchard11} focuses on a specific type of traffic, while our investigation is more systematic (for the single node case), applicable to any type of traffic that has v.b.c stochastic arrival curve, which covers a wide range of traffic types \cite{Jiang-comnet09}. 

For the examples, delay bound analysis of $M/M/1$ using snetcal can be found in \cite{Ciucu07b}\cite{Jiang-valuetools09}. However, the technique used in this paper has fundamental difference from the techniques used in \cite{Ciucu07b}\cite{Jiang-valuetools09}. Particularly in \cite{Ciucu07b}\cite{Jiang-valuetools09}, the analysis directly works on the arrival process and the service process, without mapping the arrival process to the stochastic arrival curve characterization, nor proving the stochastic service curve characterization of the system taking into consideration the packetization effect. For delay bound analysis of $M/M/1/$priority using snetcal, the same delay expectation bound as (\ref{db-exp2}) may be found in \cite{Ciucu07-thesis}. However, beside the fundamental difference in the used analytical technique, the bound in \cite{Ciucu07-thesis} is derived under some additional conditions/assumptions, e.g., preemptive priority and ignoring the packetizer. Nevertheless, it is exciting to see the same bound derived when the packetization effect is taken into account.

\section{Conclusion} \label{sec-9}

In this paper, we considered a packet-switched network node with constant capacity (in bps) and systematically derived stochastic service curves and delay bounds for the system. Specifically, we proved that the node provides a stochastic service curve with a bounding function equal to the CCDF of packet length distribution. In addition, we derived delay bounds, which imply that {\em the last packetizer can be ignored} property may be extended to SNC. Furthermore, we presented relations that allow to exploit independence information in the analysis. For the single flow case, a by-product is a new delay bound that matches with the exact result for $M/M/1$. 

Recall that, while the considered system is perhaps the simplest computer network system, before this work, in the context of stochastic network calculus, little was known about how to make use of the independence information in the analysis, particularly when the packetization effect is considered. This paper makes one step forward. We believe the analysis may be extended to the network case, where how to make use of flow independence information to improve results (without ignoring the packetization effect) still remains largely mysterious.

\bibliographystyle{abbrv}
\bibliography{nc-qt}

\begin{thebibliography}{10}

\bibitem{DiffServ}
S.~Blake and et~al.
\newblock An architecture for {D}ifferentiated {S}ervices.
\newblock {\em IETF RFC 2475}, Dec. 1998.

\bibitem{Burchard11}
A.~Burchard, J.~Liebeherr, and F.~Ciucu.
\newblock On superlinear scaling of network delays.
\newblock {\em IEEE/ACM Trans. Netw.}, 19(4):1043--1056, 2011.

\bibitem{Burchard06}
A.~Burchard, J.~Liebeherr, and S.~D. Patek.
\newblock A min-plus calculus for end-to-end statistical service guarantees.
\newblock {\em IEEE Trans. Information Theory}, 52:4105--4114, 2006.

\bibitem{Chang94}
C.~S. Chang.
\newblock Stability, queue length, and delay of deterministic and stochastic
  queueing networks.
\newblock {\em IEEE Trans. Auto. Control}, 39(5):913--931, May 1994.

\bibitem{Chang00}
C.-S. Chang.
\newblock {\em Performance Guarantees in Communication Networks}.
\newblock Springer-Verlag, 2000.

\bibitem{Ciucu07b}
F.~Ciucu.
\newblock Network calculus delay bounds in queueing networks with exact
  solutions.
\newblock In {\em Proc. ITC 2007}, 2007.

\bibitem{Ciucu07-thesis}
F.~Ciucu.
\newblock {\em Scaling Properties in the Stochastic Network Calculus}.
\newblock PhD thesis, University of Virginia, August 2007.

\bibitem{Ciucu06}
F.~Ciucu, A.~Burchard, and J.~Liebeherr.
\newblock A network service curve approach for the stochastic analysis of
  networks.
\newblock {\em IEEE Trans. Information Theory}, 52(6):2300--2312, June 2006.

\bibitem{CiucuS12}
F.~Ciucu and J.~Schmitt.
\newblock Perspectives on network calculus: no free lunch, but still good
  value.
\newblock In {\em SIGCOMM}, pages 311--322, 2012.

\bibitem{Cruz91ab}
R.~L. Cruz.
\newblock A calculus for network delay, part {I} and part {II}.
\newblock {\em IEEE Trans. Information Theory}, 37(1):114--141, Jan. 1991.

\bibitem{Cruz96}
R.~L. Cruz.
\newblock Quality of service management in integrated services networks.
\newblock In {\em Proc. 1st Semi-Annual Research Review, CWC, UCSD}, June 1996.

\bibitem{Fidler06b}
M.~Fidler.
\newblock An end-to-end probabilistic network calculus with moment generating
  functions.
\newblock In {\em Proceedings of IWQoS}, 2006.

\bibitem{Fidler10}
M.~Fidler.
\newblock A survey of deterministic and stochastic service curve models in the
  network calculus.
\newblock {\em IEEE Comm. Sur. \& Tut.}, 12(1), 2010.

\bibitem{Jiang01}
Y.~Jiang.
\newblock Delay bounds for a network of {G}uaranteed {R}ate servers with {FIFO}
  aggregation.
\newblock {\em Computer Networks}, 40(6):683--694, Dec. 2002.

\bibitem{Jiang03}
Y.~Jiang.
\newblock Relationship between guaranteed rate server and latency rate server.
\newblock {\em Computer Networks}, 43(3):307--315, 2003.

\bibitem{Sigcomm06}
Y.~Jiang.
\newblock A basic stochastic network calculus.
\newblock In {\em Proc. ACM SIGCOMM 2006}, pages 123--134, 2006.

\bibitem{Jiang-ton06}
Y.~Jiang.
\newblock Per-domain packet scale rate guarantee for expedited forwarding.
\newblock {\em IEEE/ACM Trans. Networking}, 14:630--643, 2006.

\bibitem{Jiang-valuetools09}
Y.~Jiang.
\newblock Network calculus and queueing theory: Two sides of one coin.
\newblock In {\em Proc. Valuetools}, 2009.
\newblock Updated version on 16-March-2010.

\bibitem{Jiang-note10}
Y.~Jiang.
\newblock A note on applying stochastic network calculus, May 2010.

\bibitem{Jiang12}
Y.~Jiang.
\newblock Stochastic network calculus for performance analysis of internet
  networks: An overview and outlook.
\newblock In {\em Proc. ICNC}, 2012.

\bibitem{SNetCal}
Y.~Jiang and Y.~Liu.
\newblock {\em Stochastic Network Calculus}.
\newblock Springer-Verlag, 2008.

\bibitem{Jiang-comnet09}
Y.~Jiang, Q.~Yin, Y.~Liu, and S.~Jiang.
\newblock Fundamental calculus on generalized stochastically bounded bursty
  traffic for communication networks.
\newblock {\em Computer Networks}, 53(12):2011--2021, 2009.

\bibitem{Kelly96}
F.~Kelly.
\newblock Notes on effective bandwidths.
\newblock In {\em Stochastic Networks: Theory and Applications, Royal
  Statistical Society Lecture Notes Series, 4. Oxford University Press}, 1996.

\bibitem{Kleinrock75}
L.~Kleinrock.
\newblock {\em Queueing Systems, Volume 1: Theory}.
\newblock John Wiley \& Sons, 1975.

\bibitem{KumarMK04}
A.~Kumar, D.~Manjunath, and J.~Kuri.
\newblock {\em Communication Networking: An Analytical Approach}.
\newblock Morgan Kaufmann, 2004.

\bibitem{Kurose92}
J.~Kurose.
\newblock On computing per-session performance bounds in high-speed multi-hop
  computer networks.
\newblock In {\em ACM SIGMETRICS'92}, 1992.

\bibitem{NetCal}
J.-Y. {Le Boudec} and P.~Thiran.
\newblock {\em Network Calculus: A Theory of Deterministic Queueing Systems for
  the Internet}.
\newblock Springer-Verlag, 2001.

\bibitem{Lee95}
K.~Lee.
\newblock Performance bounds in communication networks with variable-rate
  links.
\newblock In {\em Proc. ACM SIGCOMM'95}, 1995.

\bibitem{Li07}
C.~Li, A.~Burchard, and J.~Liebeherr.
\newblock A network calculus with effective bandwidth.
\newblock {\em IEEE/ACM Trans. Networking}, 15(6):1442--1453, December 2007.

\bibitem{LiebeherrBC12}
J.~Liebeherr, A.~Burchard, and F.~Ciucu.
\newblock Delay bounds in communication networks with heavy-tailed and
  self-similar traffic.
\newblock {\em IEEE Transactions on Information Theory}, 58(2):1010--1024,
  2012.

\bibitem{Liu07}
Y.~Liu, C.-K. Tham, and Y.~Jiang.
\newblock A calculus for stochastic {QoS} analysis.
\newblock {\em Performance Evaluation}, 64:547--572, 2007.

\bibitem{Mao06}
S.~Mao and S.~S. Panwar.
\newblock A survey of envelope processes and their applications in quality of
  service provisioning.
\newblock {\em IEEE Comm. Sur. \& Tut.}, 8(1-4):2--20, 2006.

\bibitem{YS93}
O.~Yaron and M.~Sidi.
\newblock Performance and stability of communication network via robust
  exponential bounds.
\newblock {\em IEEE/ACM Trans. Networking}, 1(3):372--385, June 1993.

\bibitem{LCN02}
Q.~Yin, Y.~Jiang, S.~Jiang, and P.~Y. Kong.
\newblock Analysis on generalized stochastically bounded bursty traffic for
  communication networks.
\newblock In {\em Proc. IEEE LCN'02}, 2002.

\end{thebibliography}

\section*{Appendix: Proof of Lemma \ref{lm-relation}}

Consider any time $t$ and any sample path of the system. 
Since $i(t)= \min \{k: d^{f,k} \ge t\}$, meaning $p^{f,i(t)}$ is the most recent packet of flow $f$ in $A^{*}(t)$ which departs from the node after $t$, there holds: 
\begin{equation}\label{mid-0a}
A^{*}(t) \ge \sum_{k=1}^{i(t)-1} l^{f,k}
\end{equation}
where the equality holds when $d^{f,i(t)} > t$; otherwise when $d^{f,i(t)} = t$, $A^{*}(t) > \sum_{k=1}^{i(t)-1} l^{f,k}$ since $A^{*}(t) = \sum_{k=1}^{i(t)} l^{f,k}$. 

Define 
$$
i'=\max\{k: a^{f,k} \le t \}
$$ 
which means that $p^{f, i'}$ is the last arrival packet in $A(t)$ which arrives at or before time $t$. In other words, we have
$$
a^{f,i'+1} >t.
$$
In addition, comparing $i'$ with $i(t)$, we must have
$$
i' \ge i(t) -1 
$$
because $i(t)-1$ is the last departure packet before or on time $t$, which has to have arrived before or at time $t$. 

Let us split the time period $[0,t]$ into $i'+1$ intervals, which are $[0, a^{f,1})$, \dots, $[a^{f,j-1}, a^{f,j})$, \dots, $[a^{f, i'-1}, a^{f, i'})$ and $[a^{f, i'}, t]$. Then, consider
\begin{eqnarray}
&& A\otimes\beta_h(t) - A^{*}(t) \nonumber \\
&=& \inf_{0 \le s \le t} \{A(s) + R \cdot (t-s) - A^{*}(t)\} \nonumber\\
&=& R \cdot \inf_{0 \le s \le t} \{(t-s) + \frac{A(s) - A^{*}(t)}{R}\} \nonumber \\
&=& R \cdot \min \left [ \inf_{0 \le s <a^{f,1}} \{(t-s) + \frac{A(s) - A^{*}(t)}{R}\}, \dots, \right. \nonumber\\
&& \inf_{a^{f,j-1} \le s <a^{f,j}} \{(t-s) + \frac{A(s) - A^{*}(t)}{R}\}, \dots, \nonumber\\
&& \left. \inf_{a^{f,i'} \le s \le t} \{(t-s) + \frac{A(s) - A^{*}(t)}{R}\} \right ] \label{mid-00a}
\end{eqnarray}

For the first interval, we have $A(s) = 0$ for $\forall s: 0 \le s < a^{f,1}$. Hence 
\begin{eqnarray}
&& \inf_{0 \le s < a^{f,1}} \{(t-s) + \frac{A(s) - A^{*}(t)}{R}\} \nonumber \\
&=& t- \frac{A^{*}(t)}{R} + \inf_{0 \le s < a^{f,1}} \{-s + \frac{A(s)}{R}\} \nonumber \\
&=& t- \frac{A^{*}(t)}{R} - a^{f,1} \nonumber \\
&\le& d^{f,i(t)} - \left[ a^{f,1} + \frac{A^{*}(t)}{R} \right ] 
\end{eqnarray}

Similarly, for the next $i'-1$ intervals, we have $A(s)=\sum_{k=1}^{j-1}l^{f,k}$, for $\forall s: a^{f,j-1} \le s < a^{f,j}$, $(j=2, \dots, i')$. Hence, 
\begin{eqnarray}
&& \inf_{a^{f,j-1} \le s <a^{f,j}} \{(t-s) + \frac{A(s) - A^{*}(t)}{R}\} \nonumber \\
&=& t- \frac{A^{*}(t)}{R} + \inf_{a^{f,j-1} \le s <a^{f,j}} \{-s + \frac{A(s)}{R}\} \nonumber \\
&=& t- \frac{A^{*}(t)}{R} - a^{f,j} + \frac{\sum_{k=1}^{j-1}l^{f,k}}{R} \\
&\le& d^{f,i(t)} - \left[ a^{f,j} + \frac{A^{*}(t) - \sum_{k=1}^{j-1}l^{f,k}}{R} \right ]
\end{eqnarray}

Note that for $\forall s: a^{f,i'} \le s < a^{f,i'+1}$, we have $A(s)=\sum_{k=1}^{i'}l^{f,k}$. In addition, in the above discussion, it is known $t < a^{f,i'+1}$. Hence, for the last interval, we have $A(s)=\sum_{k=1}^{i'}l^{f,k}$ for $\forall s: a^{f,i'} \le s \le t$. Consequently, 
\begin{eqnarray}
&&\inf_{a^{f,i'} \le s \le t} \{(t-s) + \frac{A(s) - A^{*}(t)}{R}\} \nonumber \\
&=& t- \frac{A^{*}(t)}{R} + \inf_{a^{f,i'} \le s \le t} \{-s + \frac{A(s)}{R}\} \nonumber \\
&=& t- \frac{A^{*}(t)}{R} - t + \frac{\sum_{k=1}^{i'}l^{f,k}}{R} \nonumber \\
&=& - \frac{A^{*}(t)}{R} + \frac{\sum_{k=1}^{i'}l^{f,k}}{R} \nonumber \\
&\le& d^{f,i(t)} -  \left[ a^{f,i(t)} + \frac{A^{*}(t)-\sum_{k=1}^{i'}l^{f,k}}{R} \right]
\end{eqnarray}

Considering all these $i'+1$ intervals, we get
{\small
\begin{eqnarray}
&& A\otimes\beta_h(t) - A^{*}(t) \nonumber \\
&\le& R \cdot d^{f,i(t)} - \nonumber \\
&& R \cdot \max [ a^{f,1} +  \frac{A^{*}(t)}{R}, \dots, a^{f,j} +  \frac{A^{*}(t) - \sum_{k=1}^{j-1}l^{f,k}}{R}, \dots,  \nonumber \\
&& a^{f,i'} +  \frac{A^{*}(t) - \sum_{k=1}^{i'}l^{f,k}}{R}, a^{f,i(t)} + \frac{\sum_{k=i'+1}^{i(t)}l^{f,k}}{R} ] \\
&\le & l^{f, i(t)} + R \cdot d^{f,i(t)} - \nonumber \\
&& R \cdot \max [ a^{f,1} + \frac{\sum_{k=1}^{i(t)} l^{f,k}}{R}, \dots, a^{f,j} + \frac{\sum_{k=j}^{i(t)} l^{f,k}}{R}, \dots, \nonumber \\
&& a^{f,i'} + \frac{\sum_{k=i'}^{i(t)} l^{f,k}}{R}, a^{f,i(t)} + \frac{\sum_{k=i'+1}^{i(t)}l^{f,k}}{R} ] \label{mid-1a}\\
&\le& l^{f, i(t)} + R \cdot d^{f,i(t)} - \nonumber \\
&& R \cdot \max \left [ a^{f,1} + \frac{\sum_{k=1}^{i(t)} l^{f,k}}{R}, \dots, a^{f,i(t)} + \frac{\sum_{k=i(t)}^{i(t)}l^{f,k}}{R} \right] \label{mid-1b} \\
&=& R \cdot [d^{f,i(t)} - V^{f, i(t)}(R)] + l^{f, i(t)}.
\end{eqnarray}
}

Here step (\ref{mid-1b}) is due to $i(t) \le i'+1$ and taking maximum on the first $i(t)$ elements of the third term in (\ref{mid-1a}) results in a smaller or equal value.

\end{document}